\documentclass[a4paper]{article}

\usepackage{amsmath,amsfonts,amssymb,amsthm}
\usepackage{mathrsfs}
\usepackage{bm}
\usepackage{extarrows}
\usepackage{geometry}
\usepackage{authblk}
\usepackage{multirow}
\usepackage{hyperref}
\usepackage[ruled]{algorithm2e}
\usepackage{graphicx}
\usepackage{caption,subcaption}
\usepackage{pifont}
\usepackage{diagbox}

\newtheorem{lemma}{Lemma}

\newtheorem{corollary}{Corollary}

\DeclareMathAccent{\widehat}{\mathord}{largesymbols}{"62}

\DeclareMathOperator{\ord}{ord}

\newcommand{\pbra}[1]{\left( #1 \right)}
\newcommand{\cbra}[1]{\left\{ #1 \right\}}

\newcommand{\Zbb}{\mathbb{Z}}

\title{A Note on Lower Digits Extraction Polynomial for Bootstrapping}
\author[1]{Mingjia Huo\thanks{mingjia@pku.edu.cn}}
\author[1]{Kewen Wu\thanks{shlw\_kevin@pku.edu.cn}}
\author[2]{Qi Ye\thanks{yeq18@mails.tsinghua.edu.cn}}
\affil[1]{School of Electronics Engineering and Computer Science, Peking University, Beijing, China}
\affil[2]{Institute for Interdisciplinary Information Sciences, Tsinghua University, Beijing, China}
\date{}

\begin{document}

\maketitle

\begin{abstract}
    Bootstrapping is a crucial but computationally expensive step for realizing Fully Homomorphic Encryption (FHE). Recently, Chen and Han (Eurocrypt 2018)
    introduced a family of low-degree polynomials to extract the lowest digit with respect to a certain congruence, which helps improve the bootstrapping for both FV and BGV schemes.
    
    In this note, we present the following relevant findings about the work of Chen and Han (referred to as CH18): 
    \begin{itemize}
        \item We provide a simpler construction of the low-degree polynomials that serve the same purpose and match the asymptotic bound achieved in CH18;
        \item We show the optimality and limit of our approach by solving a minimal polynomial degree problem;
        \item We consider the problem of extracting other low-order digits using polynomials, 
            and provide negative results.
    \end{itemize}
\end{abstract}

\section{Introduction}
Fully homomorphic encryption (FHE) \cite{G09} is a special form of encryption that allows arbitrary computation on ciphertexts, producing a ciphertext which decrypts to the result of the desired operations (as if they had been performed) on the plaintexts. 
 In a typical homomorphic encryption, each fresh ciphertext starts with a small initial``noise'' and grows with homomorphic operations until it eventually reaches a threshold and causes decryption failures. To solve this problem, {\sl Bootstrapping} is proposed as a ``refreshing'' procedure by homomorphically evaluating its own decryption algorithm (on highly noisy ciphertexts). Note that bootstrapping does not  eliminate the noise completely but only mitigates it. It is therefore important to squash the decryption circuit as shallow as possible to reduce the noise brought in
 by bootstrapping itself.

A prominent work by Halevi and Shoup~\cite{HS15} optimized and implemented the bootstrapping over the BGV scheme~\cite{BGV}.
Subsequent works show that the approach can also be applied to the FV scheme~\cite{FV}.
The bootstrapping procedure mainly consists of five steps: modulus switching, dot product, linear transform, digit extraction and ``inverse'' linear transform, of which {\sl digit extraction} is the most
 time-consuming. Chen and Han  \cite{CH18} proposed an improved digit extraction method that significantly brings down the depth and number of multiplications. However,  the work of
\cite{CH18} employs a rather complicated construction of polynomials. Therefore, it is naturally to ask the following question: ``is there any simpler, more intuitive, yet still efficient, way to realize digit extraction? ''

Towards this purpose, this note presents a simpler solution, which bears the same asymptotic bounds in both depth and number of multiplications with \cite{CH18}, and is better than the construction in the original implementation~\cite{HS15}. In addition, we discuss the optimality and limit of our approach. Furthermore, we provide a negative result on the existence of polynomials to directly extract more than one digits, which rules out the possibility to further improve the digit extraction proposed in \cite{mathoverflow}.

\section{Simpler Lowest Digit Extraction Polynomial}

For simplicity, we use $[n]$ to denote  $\cbra{0,1,\ldots,n-1}$ and $\deg(\cdot)$ refer the degree of a polynomial (i.e., the highest degree of its terms with non-zero coefficients). 
Also, we define $\ord_p(n)$ as the largest integer $e\geq0$ satisfying $p^e|n!$; and $\ord^{-1}_p(e)$ as the smallest integer
$n\geq0$ satisfying $p^e|n!$.

We say that a polynomial is of  degree $d$ if it has maximum degree at most $d$.
Unless otherwise specified, all polynomials we discuss in the following are of integral coefficients.

In \cite{HS15}, the authors constructed a special polynomial $F_e(\cdot)$ with the following lifting property.
We adopt the description of it from \cite{CH18}.

\begin{lemma}[Corollary 5.5 from \cite{HS15}]
    For every prime $p$ and $e\geq1$, there exists a degree-$p$ polynomial $F_e$ such that for every integer $z_0,z_1$
    with $z_0\in[p]$ and every $1\leq e'\leq e$, we have 
    $$
    F_e\pbra{z_0+p^{e'}z_1}\equiv z_0\pmod{p^{e'+1}}.
    $$
\end{lemma}

Composing the polynomial with itself several times yields a polynomial $G_e(\cdot)$ that extracts the lowest digit.

\begin{corollary}\label{HS15LDE}
    For every prime $p$ and $e\geq1$, there exists a degree-$p^{e-1}$ polynomial $G_e$ such that for every integer $z_0,z_1$
    with $z_0\in[p]$, we have 
    $$
    G_e\pbra{z_0+pz_1}\equiv z_0\pmod{p^e}.
    $$
\end{corollary}
\begin{proof}
    Let $G_e=\underbrace{F_e\circ\cdots\circ F_e}_{e-1}$, then it can be verified using induction on
    $$
    \underbrace{F_e\circ\cdots\circ F_e}_i(z_0+pz_1)\equiv z_0\pmod{p^{i+1}}.
    $$
\end{proof}

In \cite{CH18}, they managed to construct a polynomial with same purpose but much lower degree.

\begin{lemma}[Lemma 3 in \cite{CH18}]\label{CH18LDE}
    For every prime $p$ and $e\geq1$, there exists a degree-$((e-1)(p-1)+1)$ polynomial $H_e$ such that for every integer
    $z_0,z_1$ with $z_0\in[p]$, we have
    $$
    H_e(z_0+pz_1)\equiv z_0\pmod{p^e}.
    $$
\end{lemma}

However, their construction is complicated; and during their evaluation, the term $(e-1)(p-1)+1$ is somewhat too heavy
and they simply enlarge it to $ep$ for convenience.
Hence, we present a much simpler and more intuitive construction directly from \cite{HS15} without hurting the asymptotic 
performance of the algorithm in \cite{CH18}.

\begin{lemma}\label{newLDE}
    For every prime $p$ and $e\geq1$, there exists a degree-$(ep-1)$ polynomial $L_e$ such that for every integer
    $z_0,z_1$ with $z_0\in[p]$, we have
    $$
    L_e(z_0+pz_1)\equiv z_0\pmod{p^e}.
    $$
\end{lemma}
\begin{proof}
    Observe that for any $x\in\Zbb$
    $$
    f(x):=\pbra{x\pbra{x^{p-1}-1}}^e\equiv0\pmod{p^e}.
    $$
    This identity can be easily verified by considering whether $p$ divides $x$.
    \footnote{The idea of utilizing $(x^p-x)^e$ was suggested in \cite{mathoverflow} by Will Sawin in the comment, 
    but was not used to give the construction in \cite{CH18}.}

    Note that the coefficient of the term with highest degree in $f(x)$ is $1$.
    Then we can repeatedly subtracting multiple of $f(x)$ from $G_e$ in Lemma \ref{HS15LDE},
    as long as it has degree greater than $\deg(f)=ep$. I.e.,
    $$
    L_e=G_e\pmod{f(x)}.
    $$
\end{proof}

\section{Lowest Degree of Non-trivial Zero Polynomial}

The construction in Lemma \ref{newLDE} relies on the polynomial $\pbra{x\pbra{x^{p-1}-1}}^e$, which vanishes on every integer
after modulo $p^e$; and the coefficient of its highest degree term is $1$.
If we could find a polynomial with the same properties but lower degree, Lemma \ref{newLDE} can be further improved.

However in this section, we show this approach will fail.

\begin{lemma}\label{zeropoly}
    For every prime $p$ and $e\geq1$, assume that polynomial $f(x)$
    \begin{itemize}
        \item vanishes on every integer after modulo $p^e$, i.e., $f(x)\equiv0\pmod{p^e}$ holds for any integer $x$;
        \item is non-trivial, i.e., $\deg(f)\neq0$;
        \item has coefficient $1$ on its highest degree term.
    \end{itemize}
    Then $\deg(f)\geq\ord_p^{-1}(e)$, and this lower bound can be attained.
\end{lemma}
\begin{proof}
    Assume $f(x)$ is an arbitrary polynomial satisfying the conditions and $\deg(f)=k>0$. Then without loss of generality, 
    we can write it as
    $$
    f(x)=a_0+a_1x+a_2x(x-1)+\cdots+a_k\prod_{i=0}^{k-1}(x-i).
    $$
    Now we prove $a_ip^{\ord_p(i)}\equiv0\pmod{p^e}$ by induction on $i$.
    \begin{itemize}
        \item $i=0$. Observe that $f(0)=a_0\equiv0\pmod{p^e}$, thus the claim holds immediately.
        \item $i=j+1,j\geq0$. Observe that
            $$
            f(j+1)
            =\sum_{u=0}^{j+1}a_u\prod_{v=0}^{u-1}(j+1-v)
            =\sum_{u=0}^{j+1}a_uu!\binom{j+1}{u}
            \equiv a_{j+1}p^{\ord_p(j+1)}\equiv0\pmod{p^e},
            $$
            thus the claim holds as well.
    \end{itemize}
    Since $f(x)$ has coefficient $1$ on its highest degree term, $a_k$ should equal $1$ and thus $\deg(f)=k\geq\ord_p^{-1}(e)$.

    To show this lower bound can be actually achieved, we construct 
    $$
    f(x)=\prod_{i=0}^{d-1}(x-i)=d!\binom xd\equiv0\pmod{p^e},
    $$
    where $d=\ord_p^{-1}(e)$.
\end{proof}

Although this provides a better polynomial, its improvement is marginal as
$$
    \ord_p(d)=\sum_{i=1}^{+\infty}\left\lfloor\frac dp\right\rfloor\leq\sum_{i=1}^{+\infty}\frac dp=\frac{d}{p-1},
$$
thus $\ord_p^{-1}(e)\geq e(p-1)$.

\section{Extracting Other Lower Digits using Polynomial}

Following \cite{mathoverflow}, we consider the problem to extract other lower digits using polynomial, which can also be used to accelerate the bootstrapping process in \cite{HS15}.

However, we show that the answer is negative.

\begin{lemma}\label{rdigitremove}
    For any prime $p$ and $1<r<e$, there does not exist polynomial $f(x)$ such that for every integer $z_0,z_1$ with
    $z_0\in[p^r]$, 
    $$
    f(z_0+z_1p^r)\equiv z_1p^r\pmod{p^e}.
    $$
\end{lemma}
\begin{proof}
    Assume such $f(x)$ exists and
    $$
    f(x)=a_0+a_1x+a_2x^2+\cdots.
    $$
    Observe that
    \begin{align}
        &0\equiv f(0)\equiv a_0\pmod{p^e}\label{eq1}\\
        &p^{e-1}\equiv f(p^{e-1})\equiv a_0+a_1p^{e-1}\pmod{p^e}\label{eq2}\\
        &0\equiv f(p)\equiv a_0+a_1p+a_2p^2+\cdots\pmod{p^e}\label{eq3}.
    \end{align}
    Equation \ref{eq1} shows $a_0|p^e$. Thus in Equation \ref{eq2}, we have $p|(a_1-1)$.
    Note that Equation \ref{eq3} implies
    $$
        p^2\ |\ p^e\ |\ \pbra{a_0+a_1p+a_2p^2+\cdots},
    $$
    which gives $p|a_1$ and contradicts to $p|(a_1-1)$.
\end{proof}

As a corollary, such lower digits extraction polynomial does not exist.

\begin{corollary}\label{rdigitext}
    For any prime $p$ and $1<r<e$, there does not exist polynomial $f(x)$ such that for every integer $z_0,z_1$ with
    $z_0\in[p^r]$, 
    $$
    f(z_0+z_1p^r)\equiv z_0\pmod{p^e}.
    $$
\end{corollary}
\begin{proof}
    Assume such $f(x)$ exists. 
    Then $g(x):=x-f(x)$ contradicts the statement in Lemma \ref{rdigitremove}.
\end{proof}

Note that though $r>1$ seems odd in Corollary \ref{rdigitext}, it is inevitable since when $r=1$ we do have such polynomial
in Lemma \ref{CH18LDE} or Lemma \ref{newLDE}.

\section*{Acknowledgment}
We would like to thank Prof. Andrew C. Yao and Prof. Yu Yu for hosting the crypto study group in Tsinghua University 
and giving many helpful comments on the manuscript.
We would also like to thank other students in the crypto study group for valuable discussions.

\bibliography{ref}
\end{document}